\documentclass[a4paper,twocolumn,11pt,unpublished]{quantumarticle}
\pdfoutput=1
\usepackage[T1]{fontenc}
\usepackage[utf8]{inputenc}
\usepackage[english]{babel}
\usepackage{amsmath,amssymb,amsthm,mathtools}
\usepackage{graphicx,booktabs}
\usepackage[numbers,sort&compress]{natbib}
\usepackage{microtype}
\usepackage{hyperref}
\newcommand{\one}{\mathbf{1}}
\newcommand{\R}{\mathbb{R}}
\newcommand{\E}{\mathbb{E}}
\newcommand{\TV}{\operatorname{TV}}
\newcommand{\sech}{\operatorname{sech}}
\newcommand{\spa}{\operatorname{span}}
\newcommand{\cone}{\operatorname{cone}}
\newcommand{\Tr}{\operatorname{Tr}}
\newtheorem{theorem}{Theorem}[section]
\newtheorem{proposition}[theorem]{Proposition}

\newlength{\manuscriptcolumnwidth}
\AtBeginDocument{%
  \normalfont
  \setlength{\manuscriptcolumnwidth}{\dimexpr(\textwidth-\columnsep)/2\relax}%
  \hypersetup{allcolors=quantumviolet}%
  \fancyfoot[C]{}%
  \fancyfoot[R]{\sffamily\textcolor{quantumviolet}{\thepage}}%
}
\newcommand{\appendixfigure}[3]{%
  \begin{minipage}[t]{\manuscriptcolumnwidth}
    \vspace{0pt}\centering
    \includegraphics[width=\linewidth]{figures/#1.pdf}
    \captionof{figure}{#2}\label{#3}
  \end{minipage}%
}
\begin{document}
\title{\texorpdfstring{\textcolor{quantumviolet}{Hidden-state updates and observable-record composition in retrocausal models}}{Hidden-state updates and observable-record composition in retrocausal models}}
\author{D.M. Theshan N. Weerasinghe}
\email[\\]{theshan.dissanayakemudiyanselage1@monash.edu}
\affiliation{Department of Electrical and Computer Systems Engineering, Faculty of Engineering, Monash University, Clayton, Australia}
\date{}
\begin{abstract}
A hidden-variable model can reproduce context-independent individual readouts while failing to preserve observable records when devices are connected or reused. We distinguish two composition requirements. For a finite family of nonnegative instrument matrices, an invariant space generated by normalization defects characterizes unit continuation mass under every finite adaptive policy. Its annihilator defines the maximal admissible positive cone closed under outcome-conditioned updates. A separate observation space characterizes indistinguishability of context-dependent hidden preparations under joint readout. A finite-state example passes every fixed sequence but fails under outcome-dependent stopping. For a wrapped-Cauchy spin source with an explicitly stipulated unchanged-source passive extension, two individually fair probes acquire a setting-dependent agreement probability with contrast $\epsilon^2\sech(2g)/4$. More generally, protection under one and two identical uses forces a continuous binary passive response to be constant; measurable two-valued exceptions remain possible. Exact rational calculations, exhaustive finite-policy experiments, continuous source sampling, and quadrature test the stated examples. The criteria apply to specified instrument families and accessible records; they do not construct a microscopic detector or establish consistency under unrestricted network composition.
\end{abstract}
\maketitle

\section{Introduction}
An outcome distribution that is insensitive to a later experimental choice need not remain insensitive when its detector is reused. The first interaction changes the hidden state presented to the next interaction, and two individually uninformative outputs can contain information in their correlations. A composition criterion must therefore specify both the hidden update and the joint record that an observer can access.

This issue is particularly relevant to retrocausal models, where the hidden distribution can depend on later measurement settings. Such dependence is not itself an observable signal. The operational question is whether an earlier accessible record depends on a later freely selectable choice after all unobserved outcomes have been summed. In this work an observable record includes all declared outputs, including retained pointer, timing, controller, memory, and environment records. No postselection is used.

The distinction between hidden dependence and operational independence also appears in causal-model analyses of quantum correlations. Wood and Spekkens show that reproducing nonsignalling Bell correlations within their classical directed-acyclic causal-model framework requires fine-tuning~\cite{WoodSpekkens2015}. Leifer and Pusey distinguish operational time symmetry from additional assumptions about the underlying ontology~\cite{LeiferPusey2017}. The present analysis examines whether a specified model preserves operational independence under composition of its readouts.

Two requirements arise. Appending an apparatus must not reweight an existing archive through a context-dependent continuation mass. In addition, readout operations must not distinguish hidden preparations associated with different future contexts. The first requirement concerns normalization; the second concerns observable equivalence. We formulate both on a fixed hidden state containing every variable required for subsequent transition weights. Archived labels are immutable classical records, although they may control later settings. Describing each outcome together with its state update follows the instrument approach to sequential measurement~\cite{DaviesLewis1970}. Our finite-state kernels extend this instrument representation to nonnegative raw weights whose normalization is tested on a specified class of hidden preparations.

The finite-state analysis uses familiar invariant-span methods. Backward spaces of word effects are standard in equivalence algorithms for weighted and probabilistic automata~\cite{Kiefer2011}. Here the observation space is such a word-effect space, while a defect-seeded invariant space isolates the normalization requirement relevant to adaptive continuation. The contribution is their operational separation, the associated admissible-cone characterization, and their application to reuse, reset, readable refinements, and a repeated-probe obstruction.

The continuous example uses the endpoint weights of the Schulman-type spin construction discussed by Almada et al.~\cite{Almada2016}. Their construction treats measurement as a boundary condition; an additional intermediate measurement can change the experiment. We therefore state the unchanged-source passive-probe extension as an additional assumption. The resulting obstruction tests that extension. Wharton and Argaman discuss adaptive conditional reconstruction~\cite{Wharton2020}, which likewise must be distinguished from multiplying one fixed collection of raw kernels under every controller policy.

We permit arbitrary finite causal decision trees, with settings determined only by records already available and with stopping independently available on every branch. All outcomes, including no-click events and readable reset-environment outputs, belong to the declared instrument. These conditions define the mathematical class being tested; they do not assert physical realizability of every matrix family. Sections~\ref{sec:finite} and~\ref{sec:spin} give the finite-state criteria and continuous obstruction, respectively. Section~\ref{sec:experiments} reports numerical experiments, and Section~\ref{sec:discussion} discusses quantum calibration and the remaining physical assumptions.

\section{Composition with hidden-state updates}\label{sec:finite}
\subsection{Adaptive normalization}
Nonnegative hidden measures are represented by row vectors, and hidden-state functions by column vectors. On $n$ hidden states, setting $u$ has finitely many outcome matrices
\begin{equation}
 K_{u,o}\in\R_{\geq0}^{n\times n},\qquad T_u=\sum_o K_{u,o}.
 \label{eq:instrument}
\end{equation}
The entry $K_{u,o}(i,j)$ is the raw weight of outcome $o$ and transition $i\to j$. It need not be stochastic. A branch sends $\mu$ to $\mu K_{u,o}$; when its mass is positive the conditional update is
\begin{equation}
 \mu'_{u,o}=\frac{\mu K_{u,o}}{\mu K_{u,o}\one}.
 \label{eq:update}
\end{equation}
For a finite adaptive policy $\pi$, define its continuation effect recursively by
\begin{equation}
 e_{\rm stop}=\one,\qquad e_\pi=\sum_o K_{u,o}e_{\pi_o}.
 \label{eq:effect}
\end{equation}
Different outcomes may select different subtrees, including stopping. Its integrated continuation mass is $\mu e_\pi$.

Let an earlier archive $H=h$ have baseline probability $p_h>0$ and conditional normalized hidden law $\mu_h$. Under multiplication of raw history weights followed by one global normalization, branch-dependent continuations give
\begin{equation}
 P_\pi(h)=\frac{p_h\mu_he_{\pi_h}}{\sum_jp_j\mu_je_{\pi_j}}.
 \label{eq:archive}
\end{equation}
A common continuation mass suffices for one specified policy family. The following statement identifies when unit mass is necessary.

\begin{proposition}\label{prop:unit}
Assume the archive has at least two supported values and the controller can apply any allowed continuation on one selected value while stopping on all others. Its distribution is preserved under all such controls if and only if $\mu_he_\pi=1$ for every $h$ and $\pi$.
\end{proposition}
\begin{proof}
Sufficiency follows from Eq.~\eqref{eq:archive}. For necessity, fix $h$, set $m=\mu_he_\pi$, and assign the stopping operation to every other archive value. The identity $p_hm/(p_hm+1-p_h)=p_h$, with $0<p_h<1$, forces $m=1$.
\end{proof}
Stopping has its specified physical weight; it cannot be introduced by changing normalization conventions. Without branchwise stopping common nonunit factors can cancel. An archive with only one supported value imposes no such necessity.

Define the one-step defects and their invariant span by
\begin{align}
 d_u&=T_u\one-\one,\qquad W_0=\{0\},\nonumber\\
 W_{k+1}&=\spa\bigl(\{d_u\}_u\cup\{K_{u,o}w:w\in W_k\}_{u,o}\bigr),\nonumber\\
 W&=\bigcup_{k\geq0}W_k.
 \label{eq:W}
\end{align}
These nested spaces stabilize in finite dimension. Defect vectors can have either sign; positivity is imposed on the transition weights.

\begin{theorem}[Adaptive mass criterion]\label{thm:mass}
For every nonnegative row measure $\mu$,
\begin{equation}
 \mu e_\pi=\mu\one\ \text{for all finite }\pi
 \quad\Longleftrightarrow\quad \mu W=0.
 \label{eq:masscriterion}
\end{equation}
At each depth,
\begin{equation}
 W_k=\spa\{e_\pi-\one:\operatorname{depth}(\pi)\leq k\}.
 \label{eq:depthspan}
\end{equation}
\end{theorem}
\begin{proof}
The depth-zero statement is immediate. At the next depth,
\begin{equation}
 e_\pi-\one=d_u+\sum_oK_{u,o}(e_{\pi_o}-\one).
 \label{eq:induction}
\end{equation}
Thus every policy defect lies in the space in Eq.~\eqref{eq:W}. Conversely $d_u$ is the defect of applying $u$ and stopping after every outcome. For a fixed outcome $o$, assign a subtree $\sigma$ to that outcome and stop after every other outcome. Subtracting the effect of applying $u$ and stopping everywhere gives $K_{u,o}(e_\sigma-\one)$. Differences of effects are differences of defects. Induction proves Eq.~\eqref{eq:depthspan}; pairing its union with $\mu$ proves the equivalence.
\end{proof}
The admissible cone and its normalized slice are
\begin{align}
 C_* &=\{\mu\geq0:\mu W=0\},\nonumber\\
 C_{*,1}&=\{\mu\in C_*:\mu\one=1\}.
 \label{eq:cone}
\end{align}
The normalized slice can be empty.

\begin{theorem}[Closure and maximality]\label{thm:cone}
For $\mu\in C_*$, every outcome image $\mu K_{u,o}$ lies in $C_*$ and $\mu T_u\one=\mu\one$. Every positive-mass conditional update therefore belongs to $C_{*,1}$. Moreover, any cone $C$ of nonnegative measures satisfying
\begin{equation}
 CK_{u,o}\subseteq C,\qquad \mu d_u=0\quad(\mu\in C)
 \label{eq:closedcone}
\end{equation}
is contained in $C_*$.
\end{theorem}
\begin{proof}
For $w\in W$, invariance gives $(\mu K_{u,o})w=\mu(K_{u,o}w)=0$; nonnegativity follows from the matrix entries. The mass identity is $\mu d_u=0$. Normalization uses the conic property. If Eq.~\eqref{eq:closedcone} holds, every successive branch image of $\mu\in C$ annihilates each defect. Hence $\mu$ annihilates every iterated image $K_{u_1,o_1}\cdots K_{u_m,o_m}d_v$, whose span is $W$.
\end{proof}
Applying Theorem~\ref{thm:mass} to each remaining subtree protects every generated record prefix, as well as the original archive, when its conditional law is admissible. Pointwise stochasticity, $T_u\one=\one$, is the special case $W=0$ with all nonnegative laws admissible.

A basis is computed by retaining independent defects and repeatedly applying every outcome matrix to newly admitted vectors. At most $n$ vectors are admitted. If the defects vanish, the process stops immediately; otherwise every nonstationary step increases dimension. Consequently $W_n=W$, and any failure for a specified $\mu$ has a policy witness of depth at most $n$. The witness may require different actions after different outcomes. Exact rational rank tests provide finite certificates for rational matrices; floating-point rank thresholds alone do not. This dimension bound presumes a fixed hidden-state description.

\subsection{Record indistinguishability}
Suppose normalized hidden laws $\rho_y$ depend on a later context $y$, and the same readout instruments are available in every context. Define
\begin{align}
 V_0&=\spa\{\one\},\nonumber\\
 V_{k+1}&=\spa\bigl(V_k\cup\{K_{u,o}v:v\in V_k\}_{u,o}\bigr),\nonumber\\
 V&=\bigcup_{k\geq0}V_k.
 \label{eq:V}
\end{align}
This observation space contains all branch effects. Unlike $W$, it starts from the unit effect. In fact $W\subseteq V$, since every defect is a sum of one-step branch effects minus $\one$; their distinct seeds nevertheless give distinct operational tests.

\begin{theorem}[Record indistinguishability]\label{thm:observe}
Assume $\rho_y\in C_{*,1}$ for every context. Joint record distributions agree under every finite adaptive readout policy if and only if
\begin{equation}
 (\rho_y-\rho_{y'})V=0\qquad\text{for every }y,y'.
 \label{eq:observe}
\end{equation}
\end{theorem}
\begin{proof}
For a branch word $w=(u_1,o_1)\cdots(u_m,o_m)$, write $K_w=K_{u_1,o_1}\cdots K_{u_m,o_m}$. Its probability is $\rho_yK_w\one$. Word effects span $V$, so equality of all fixed-sequence record probabilities implies Eq.~\eqref{eq:observe}. Conversely each adaptive leaf is a word, and the condition gives equality leaf by leaf and hence for every event. Membership in $C_{*,1}$ makes these raw branch weights normalized probabilities.
\end{proof}
Without the admissible-slice assumption the word argument characterizes equality of raw weights, which need not coincide with equality after context-dependent global normalization. Since $V_0$ has dimension one, $V$ stabilizes by word length $n-1$; a nonzero source difference on $V$ has a distinguishing word no longer than this.

For example, on states $(s,t)=(++,+-,-+,--)$ take
\begin{equation}
 \rho_0=(1/2,0,0,1/2),\quad \rho_1=(0,1/2,1/2,0).
 \label{eq:fairbits}
\end{equation}
Either bit is read using diagonal indicator matrices that leave the state unchanged. Each instrument sums to the identity, so $W=0$. Either bit is fair in either context, while $P(s=t)$ equals one in context 0 and zero in context 1. The one-step space has dimension three; its closure has dimension four and detects this difference.

An environmental refinement replaces $K_{u,o}$ by matrices $K_{u,o,e}$ satisfying $\sum_eK_{u,o,e}=K_{u,o}$. Invariance after summing $e$ does not imply invariance under each refined matrix. Both spaces must be recomputed when $e$ becomes readable. The accessible record and the controller's access to that record are therefore part of the specification, rather than properties inferred from a coarse instrument.

\subsection{Composition and stability}
Let $q=(1/2,1/2)$ and consider the one-outcome matrix
\begin{equation}
 D=\begin{pmatrix}3/2&0\\0&1/2\end{pmatrix}.
 \label{eq:D}
\end{equation}
Although $qD\one=1$, its normalized output is $(3/4,1/4)$, and $qD^2\one=5/4$. Two uses on branch 0 of an independent fair archive, with stopping on branch 1, give $P(H=0)=5/9$. The vectors $d=(1/2,-1/2)^T$ and $Dd=(3/4,-1/4)^T$ span $\R^2$, excluding $q$ from an indefinitely reusable class.

A nonstochastic family can nevertheless close on a nontrivial cone. For nonnegative columns $r_{u,o}$, define
\begin{equation}
 K_{u,o}=r_{u,o}q,\qquad q\sum_or_{u,o}=1.
 \label{eq:rankone}
\end{equation}
Every positive-mass branch resets its hidden law to $q$. In the explicit two-setting example,
\begin{align}
 r_{A,0}&=(1,1/4)^T,&r_{A,1}&=(1/2,1/4)^T,\nonumber\\
 r_{B,0}&=(1/4,1)^T,&r_{B,1}&=(1/4,1/2)^T.
 \label{eq:rcolumns}
\end{align}
Both settings have outcome probabilities $(5/8,3/8)$ on $q$, but $T_A\one=(3/2,1/2)^T$ is nonconstant. Here $W=\spa\{(1,-1)^T\}$ and $C_* =\cone\{q\}$. This is an algebraic existence example, not a microscopic derivation of a detector or of Bell correlations.

The stochastic reset
\begin{equation}
 S=\begin{pmatrix}1&0\\1&0\end{pmatrix}
 \label{eq:reset}
\end{equation}
sends $q$ outside the admissible slice and gives $qST_A\one=3/2$. Reset followed by $A$ on one branch of a fair archive yields $P(H=0)=3/5$. The combined defect space has dimension two. Refining $S$ with an environmental output that records the incoming state preserves this failure: the refined matrices sum to $S$ and the combined space remains two-dimensional. A normalized reset appended last preserves total mass; the failure requires subsequent reuse by a nonstochastic device. If every later device is pointwise stochastic this mass-reweighting mechanism is absent.

Adaptive control can detect failures that no fixed sequence detects. Add the readable identity instrument $Q_0=\operatorname{diag}(1,0)$, $Q_1=\operatorname{diag}(0,1)$ to setting $A$. Since $q(Q_0+Q_1)=q$ and $qT_A=q$, every fixed sequence of these settings has mass one at every length. Nevertheless, reading $Q$, applying $A$ after outcome 0, and stopping after outcome 1 gives
\begin{equation}
 qQ_0T_A\one+qQ_1\one=\frac54.
 \label{eq:adaptivewitness}
\end{equation}
Reversing the stopping rule gives $3/4$. The outcome-wise defect space has dimension two. Thus averaging over a readable refinement before analyzing the controller can miss an operationally relevant failure.

For independent systems with admissible cones $C_A,C_B$, define
\begin{equation}
 C_{AB}^{\rm prod}=\cone\{\mu\otimes\nu:\mu\in C_A,\nu\in C_B\}.
 \label{eq:prodcone}
\end{equation}
Local outcome maps $K_{u,o}\otimes I$ and $I\otimes L_{v,r}$, as well as product instruments, preserve this cone. Each product term preserves total mass; finite sums and classical feed-forward therefore also preserve it. This establishes closure for independent parallel composition and mixtures of admissible products under the stated local operations.

Admissible marginals alone do not suffice. The joint law with $\nu(0,0)=\nu(1,1)=1/2$ has both marginals $q$, yet
\begin{equation}
 \nu(T_A\one\otimes T_A\one)
 =\tfrac12[(3/2)^2+(1/2)^2]=\tfrac54,
 \label{eq:correlated}
\end{equation}
whereas $q\otimes q$ gives one. Correlated environments must be retained in the joint state. General correlated inputs require $W_{AB}$ and, if preparations depend on future contexts, $V_{AB}$ for the actual lifted and joint kernels. An entanglement-swapping operation requires its own joint kernel, outcome set, source assignment, and record specification; the criteria test such data once supplied.

For finite-horizon error control, suppose every archive-branch mass lies in $[a,b]$ with $0<a\leq b<\infty$. For $a<b$, writing $M=\sum_hp_hm_h$, convexity of $|m-M|$ on $[a,b]$ gives
\begin{align}
 \TV(P_\pi,p)&\leq\frac{(b-M)(M-a)}{M(b-a)}\nonumber\\
 &\leq\frac{\sqrt b-\sqrt a}{\sqrt b+\sqrt a}.
 \label{eq:TVbound}
\end{align}
Here $\TV(P,p)=\frac12\sum_h|P(h)-p(h)|$. The second maximum is at $M=\sqrt{ab}$; a two-point distribution assigning probability $\sqrt a/(\sqrt a+\sqrt b)$ to mass $b$ attains it. For $a=b$ the error is zero.

If every reachable one-step conditional total mass is independently bounded by $1\pm\delta$, where $0\leq\delta<1$, induction on positive decision trees bounds every depth-$N$ mass by $(1\pm\delta)^N$, with stopped branches assigned mass one. Hence
\begin{equation}
 \TV(P_\pi,p)\leq\tanh\!\left[\frac N4\log\frac{1+\delta}{1-\delta}\right].
 \label{eq:finiteerror}
\end{equation}
This premise concerns all reachable conditional states. Initial means or small basis residuals alone do not provide the requisite uniform reachability and amplification control.

\section{A passive-readout obstruction for a spin source}\label{sec:spin}
\subsection{Source and probe statistics}
Consider antiparallel hidden orientations $\lambda$ and $\lambda+\pi$, with $\lambda\in[0,2\pi)$, planar terminal axes $x,y$, and terminal outcomes $a,b\in\{+1,-1\}$. For $g>0$ define
\begin{align}
 p_g(t)&=\frac{\sinh g}{2\pi(\cosh g-\cos t)},\nonumber\\
 R_g(t)&=p_g(t)+p_g(t+\pi)\nonumber\\
 &=\frac{\sinh(2g)}{\pi[\cosh(2g)-\cos(2t)]}.
 \label{eq:wrapped}
\end{align}
The density $p_g$ integrates to one; $R_g$ integrates to two. Define $\theta_{z,a}=z+\pi\mathbf{1}_{a=-1}$. The endpoint-weight prescription~\cite{Almada2016} assigns the joint history density a weight proportional to
\begin{equation}
 p_{g_A}(\theta_{x,a}-\lambda)\,
 p_{g_B}(\theta_{y,b}-\lambda-\pi).
 \label{eq:endpoints}
\end{equation}
Summing both terminal outcomes produces
\begin{align}
 \rho_{xy}(\lambda)&=\frac{R_{g_A}(\lambda-x)R_{g_B}(\lambda-y)}{Z_{xy}},\nonumber\\
 Z_{xy}&=2R_{g_A+g_B}(x-y).
 \label{eq:rho}
\end{align}
This positive, $\pi$-periodic density is a hidden-angle marginal conditioned on the later settings. It is not a future-independent preparation. Simultaneously shifting $\lambda$ by $\pi$ and reversing both terminal outcomes preserves Eq.~\eqref{eq:endpoints}, making either terminal marginal fair.

The stipulated passive extension leaves both $\lambda$ and Eq.~\eqref{eq:rho} unchanged when probes are inserted. It uses fresh randomness independent between probes conditional on $\lambda$, with binary response
\begin{equation}
 L_\epsilon(r\mid\lambda)=\frac{1+r\epsilon\cos\lambda}{2},
 \qquad 0<\epsilon\leq1.
 \label{eq:probe}
\end{equation}
Passivity alone does not entail these source and noise assumptions. Each probe is fair because $\cos\lambda$ is antipodally odd. For two identical probes the joint earlier record is
\begin{equation}
 P(r,s\mid x,y)=\int_0^{2\pi}\rho_{xy}(\lambda)
 L_\epsilon(r\mid\lambda)L_\epsilon(s\mid\lambda)\,d\lambda.
 \label{eq:pairintegral}
\end{equation}
For equal widths $g_A=g_B=g$, the second circular moment is
\begin{equation}
 \E_{xy}e^{2i\lambda}=\frac{e^{2ix}+e^{2iy}}{2\cosh(2g)}.
 \label{eq:moment}
\end{equation}
Appendix~\ref{app:moment} derives this identity and its unequal-width extension directly from the endpoint kernels. At $x=0$,
\begin{equation}
 \E_{0y}\cos^2\lambda=\tfrac12[1+\sech(2g)\cos^2y].
 \label{eq:cos2}
\end{equation}
Consequently $P(r,s)=\frac14[1+rs\epsilon^2\E\cos^2\lambda]$.

\begin{proposition}\label{prop:spin}
In the stipulated equal-width extension with $x=0$,
\begin{align}
 P(r=+\mid y)&=P(s=+\mid y)=\tfrac12,\nonumber\\
 P(r=s\mid y)&=\tfrac12+\frac{\epsilon^2}{4}
 [1+\sech(2g)\cos^2y].
 \label{eq:agreement}
\end{align}
The agreement contrast and the total variation distance between the full pair tables at $y=0$ and $y=\pi/2$ both equal
\begin{equation}
 \Delta_{\rm agree}=\frac{\epsilon^2}{4}\sech(2g)>0.
 \label{eq:contrast}
\end{equation}
\end{proposition}
\begin{proof}
Given $\lambda$, agreement has probability $(1+\epsilon^2\cos^2\lambda)/2$. Integration using Eq.~\eqref{eq:cos2} gives the stated probability. Each of the two diagonal pair probabilities changes by $\Delta_{\rm agree}/2$, and each off-diagonal probability changes by its negative. Summing absolute differences and dividing by two gives Eq.~\eqref{eq:contrast}.
\end{proof}
For $g=0.2$, $\epsilon=0.2$, agreement is $0.5192500745$ at $y=0$ and $0.5100000000$ at $y=\pi/2$. Both terminal outcomes have been summed; the signal concerns an unselected earlier pair record. If the later axis is freely chosen after that record is archived, the stipulated extension violates its independence from this choice. The mechanism is multiplication of individually protected odd responses, which exposes an even harmonic of the source.

\subsection{Passive-readout rigidity}
Fix $x$ and positive finite $g_A,g_B$. Let a binary response be $L(+\mid\lambda)=l(\lambda)\in[0,1]$. Repetitions use the same response, unchanged hidden angle and source, and independent randomness conditional on $\lambda$.

\begin{theorem}[Two-use passive-readout rigidity]\label{thm:rigidity}
Suppose the one-probe probability is independent of every later axis $y$, with $P(+)=\ell$. Then, almost everywhere,
\begin{equation}
 l(\lambda)=\ell+f(\lambda),\qquad f(\lambda+\pi)=-f(\lambda).
 \label{eq:odd}
\end{equation}
The joint two-probe law is independent of every $y$ if and only if $f(\lambda)^2=c^2$ almost everywhere for a constant $c\geq0$. If $l$ is continuous on the circle, these conditions force $l=\ell$ everywhere. For measurable responses, nonconstant exceptions $f=c\sigma$, with antipodally odd $\sigma\in\{-1,+1\}$, are protected under every finite number of identical repetitions.
\end{theorem}
\begin{proof}
Write $A(\lambda)=R_{g_A}(\lambda-x)>0$. The Fourier series
\begin{equation}
 R_g(t)=\frac1\pi\sum_{k\in\mathbb Z}e^{-2g|k|}e^{2ikt}
 \label{eq:Fourier}
\end{equation}
has a nonzero coefficient at every even harmonic. Independence of the expectation of a bounded measurable $F$, with value $c_F$, is equivalent to
\begin{equation}
 \int_0^{2\pi}A(\lambda)[F(\lambda)-c_F]
 R_{g_B}(\lambda-y)\,d\lambda=0\quad\forall y.
 \label{eq:convolutionzero}
\end{equation}
Fourier uniqueness shows that the antipodal average of $A(F-c_F)$ vanishes almost everywhere. Since $A$ is positive and $\pi$-periodic,
\begin{equation}
 \frac{F(\lambda)+F(\lambda+\pi)}2=c_F
 \quad\text{almost everywhere}.
 \label{eq:average}
\end{equation}
Applying this to $F=l$ yields Eq.~\eqref{eq:odd}. Applying it to the two-probe event $(+,+)$, for which $F=l^2$, yields the constant antipodal average $\ell^2+f^2$. The remaining pair probabilities are determined by $\E l$ and $\E l^2$, establishing necessity and sufficiency.

If continuous $f$ satisfies $f^2=c^2>0$, connectedness forces a single sign, contradicting antipodal oddness. Thus $c=0$, and continuity extends the almost-everywhere identity to every point. For measurable $f=c\sigma$, the sign sets have equal probability under every antipodally invariant source. The condition $l\in[0,1]$ requires $c\leq\min(\ell,1-\ell)$. Every specified ordered string of $m$ outcomes containing $k$ plus signs has probability
\begin{align}
 &\tfrac12(\ell+c)^k(1-\ell-c)^{m-k}\nonumber\\
 &\quad+\tfrac12(\ell-c)^k(1-\ell+c)^{m-k},
 \label{eq:measurable}
\end{align}
which is independent of $y$.
\end{proof}
The measurable exception is necessary: an unrestricted constant-response conclusion would be false. Moreover, smooth functions can approximate a sign response, so the theorem does not supply a positive lower bound on the signal of every nonconstant continuous detector. The theorem concerns the source after both terminal outcomes are summed. Conditioning on a terminal outcome, allowing probe-induced source changes, correlated probe randomness, or context-dependent detector responses requires a different calculation.

\section{Numerical experiments}\label{sec:experiments}
The numerical study separates exact finite-state experiments, continuous-source Monte Carlo sampling, and deterministic quadrature. The simulators implement the specified instruments and endpoint laws. Analytical probabilities are used for comparison, rather than as a substitute for generating hidden states and detector outcomes. The accompanying source package contains the sampling code, retained trial records, tabulated results, and diagnostic discrepancies.

\subsection{Finite-state experiments}
Exact rational arithmetic computes $W$, $V$, and their matrix-invariance certificates. For a two-setting, binary-outcome controller with stopping, the number of syntactic policies of depth at most $d$ obeys $N_0=1$ and $N_d=1+2N_{d-1}^2$. Thus $N_3=723$, including policies that have identical effects but different decision trees. Every one of these policies has mass exactly one on $q$ for the closed family in Eq.~\eqref{eq:rcolumns}.

For the readable-state family in Eq.~\eqref{eq:adaptivewitness}, 304 of the 723 policies have nonunit mass; the extrema are $3/4$ and $5/4$. All 127 fixed setting sequences of length at most six have unit mass. The identity $qT_u=q$ proves the latter statement for all lengths independently of this enumeration. Exact dynamic programming extends the adaptive envelope to depth four, giving $[11/16,23/16]$; see Fig.~\ref{fig:adaptive}. For an unnormalized input $\mu$, the maximum is computed recursively as
\begin{equation}
 M_d(\mu)=\max\left\{\mu\one,\max_u\sum_oM_{d-1}(\mu K_{u,o})\right\},
 \label{eq:dp}
\end{equation}
with $M_0(\mu)=\mu\one$; replacing both maxima by minima gives the minimum. This optimizes over all admissible finite decision trees without individual enumeration at depth four. Attaining policies and their leaf masses are retained.

To test sensitivity to departure from the admissible input class, augment $A,B$ with $S_\eta=(1-\eta)I+\eta S$. We use $\eta\in\{0,0.01,0.05,0.1,0.2,0.5,1\}$ and horizons one through four. Direct calculation gives $qS_\eta T_A\one=1+\eta/2$. The defect-space dimension rises from one at $\eta=0$ to two at every tested positive value. The size of a finite-horizon operational error nevertheless varies continuously with $\eta$. Figure~\ref{fig:reset} shows the minimum and maximum archive probabilities when the optimized continuation is applied on branch 0 of a fair archive and branch 1 stops. At $\eta=1$ and depth four the mass range is $[1/4,9/4]$, giving archive probabilities $[1/5,9/13]$.

\begin{figure}[t]
\centering
\includegraphics[width=\columnwidth]{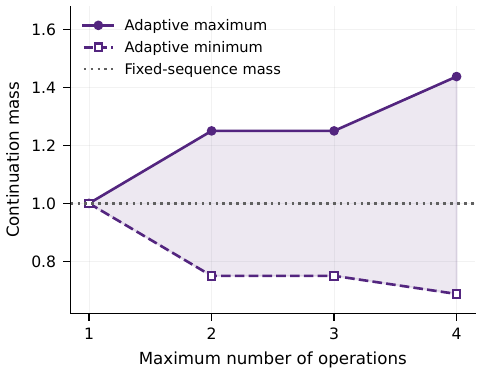}
\caption{Extremal continuation masses for the readable-state family $\{Q,A\}$ as a function of the maximum number of operations. Filled circles and the solid line denote maxima; open squares and the dashed line denote minima. The shaded region delineates the extrema, and the dotted line denotes the unit mass of every fixed sequence. Exact decision-tree optimization gives $[11/16,23/16]$ at depth four. Connecting lines are guides between integer horizons; no statistical uncertainty is present.}
\label{fig:adaptive}
\end{figure}

\begin{figure}[t]
\centering
\includegraphics[width=\columnwidth]{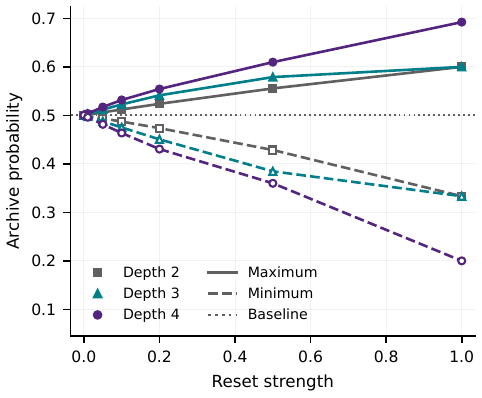}
\caption{Extremal archive probabilities for the family $\{A,B,S_\eta\}$, with $S_\eta=(1-\eta)I+\eta S$. The ordinate is $P(H=0)=m/(1+m)$, where $m$ is the optimized continuation mass on branch 0 of a fair archive and branch 1 stops. Colour and marker shape distinguish depths two, three, and four; solid and dashed lines connect exact maxima and minima, respectively. The dotted line is the baseline probability $1/2$. Depth one also gives $1/2$ for every reset strength and is omitted as redundant. Lines between sampled strengths are guides.}
\label{fig:reset}
\end{figure}

The reuse, refined-reset, correlated-product, and fair-bit examples reproduce their exact values and ranks. Three endpoint intervals independently attain Eq.~\eqref{eq:TVbound}; the archive-weight scan contains 3,003 points. These computations test examples and finite-horizon sensitivity. The all-depth conclusions follow from the invariant-space proofs and exact closure certificates, not extrapolation from a finite number of policies.

\subsection{Continuous and quantum experiments}
The primary Monte Carlo study contains 3.6 million independent trials across 23 conditions. Each retained trial contains a continuous $\lambda$, terminal outcomes $a,b$, and earlier probe records $r,s$. The axis sweep uses $g=0.2$, $\epsilon=0.6$, $y/\pi\in\{0,1/8,1/4,3/8,1/2\}$, and 200,000 trials per condition. Strength tests add $\epsilon\in\{0.2,0.4,0.8,1\}$ at $y=0,\pi/2$, also with 200,000 trials per condition. Width tests use $g\in\{0.05,0.1,0.5,1,2\}$, $\epsilon=0.6$, and the same two axes, with 100,000 trials per condition. All use $x=0$.

Sampling uses the proposal $R_g(\lambda-x)/2$. Draw independent uniform $U$ and fair $B\in\{0,1\}$, set
\begin{align}
 t&=\arctan[\tanh(g)\tan\pi(U-1/2)],\nonumber\\
 \lambda&=x+t+\pi B\pmod{2\pi},
\end{align}
and accept with probability $R_g(\lambda-y)/[\coth(g)/\pi]$. The accepted law is Eq.~\eqref{eq:rho}, with no angular discretization. Conditional terminal probabilities are the corresponding $p_g$ endpoint weight divided by the sum over that arm's two outcomes. Fresh conditional Bernoulli draws generate the probes. Independent PCG64 streams, derived from a recorded SeedSequence, separate source, terminal, and probe randomness.

Figures~\ref{fig:agreement} and~\ref{fig:marginals} compare the sampled agreement and individual-probe probabilities with their analytical values. Figure~\ref{fig:contrast} contains one contrast for each of the five strengths $\epsilon=0.2,0.4,0.6,0.8,1$, plotted at $\epsilon^2$; each contrast combines the $y=0$ and $y=\pi/2$ conditions. Binomial probabilities have 95\% Wilson intervals~\cite{Wilson1927}. Differences from independent conditions use $\widehat\Delta\pm1.96s_\Delta$, where
\begin{equation}
 s_\Delta^2=\frac{\hat p_0(1-\hat p_0)}{n_0}+\frac{\hat p_1(1-\hat p_1)}{n_1}.
\end{equation} These are pointwise intervals, not simultaneous coverage guarantees. At $g=\epsilon=0.2$, the measured contrast is $0.009025$, with interval $[0.005927,0.012123]$, compared with $0.009250$ analytically.

\begin{figure}[t]
\centering
\includegraphics[width=\columnwidth]{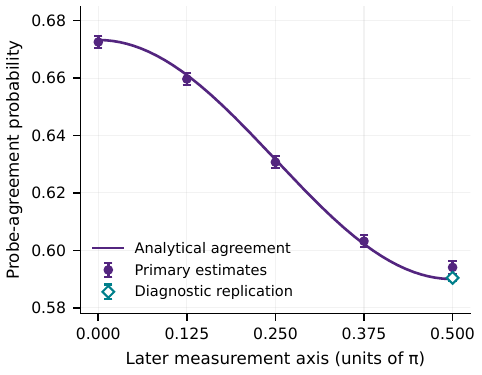}
\caption{Earlier probe-agreement probability as a function of the later axis, with abscissa $y/\pi$, at $x=0$, $g=0.2$, and $\epsilon=0.6$. The curve is $P(r=s\mid y)=1/2+\epsilon^2[1+\sech(2g)\cos^2y]/4$. Filled circles denote primary estimates from 200,000 trials per axis. The open diamond denotes the separately seeded diagnostic replication at $y=\pi/2$, with one million trials. Error bars are pointwise 95\% Wilson intervals; the diagnostic estimate does not replace the primary estimate.}
\label{fig:agreement}
\end{figure}

\begin{figure}[t]
\centering
\includegraphics[width=\columnwidth]{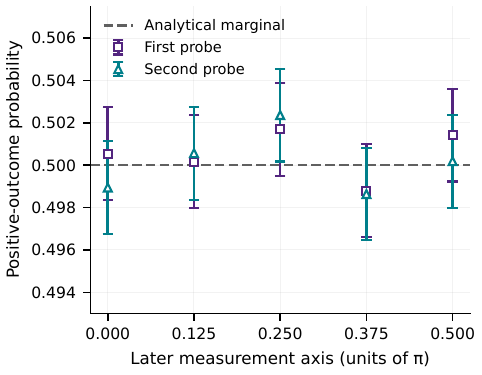}
\caption{Individual-probe probabilities for the same five primary conditions as Fig.~\ref{fig:agreement}, with abscissa $y/\pi$. Squares and triangles denote the estimates of $P(r=+\mid y)$ and $P(s=+\mid y)$, respectively, with pointwise 95\% Wilson intervals. The dashed line denotes the common analytical value $1/2$. A separate vertical scale resolves the uncertainty of the marginal estimates. The diagnostic replication is not included.}
\label{fig:marginals}
\end{figure}

\begin{figure}[t]
\centering
\includegraphics[width=\columnwidth]{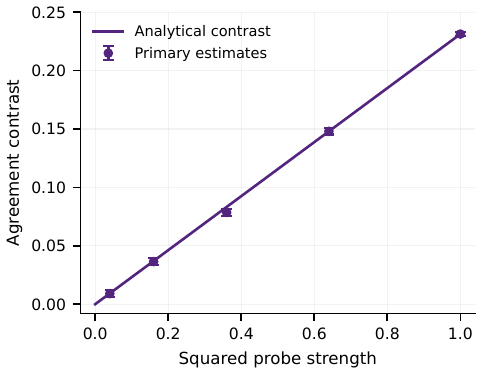}
\caption{Agreement contrast between $y=0$ and $y=\pi/2$ at $x=0$ and $g=0.2$, plotted against $\epsilon^2$. The curve is $\Delta_{\rm agree}=\epsilon^2\sech(2g)/4$. Each of the five points combines two independent primary conditions at one strength, with pointwise 95\% difference intervals. The $\epsilon=0.6$ point uses the primary samples in Fig.~\ref{fig:agreement}; the other strengths use separate simulations. The diagnostic replication is excluded.}
\label{fig:contrast}
\end{figure}

The primary orthogonal-axis point at $\epsilon=0.6$ gives agreement $0.594070$, compared with $0.590000$, a deviation of 3.70 theoretical binomial standard errors. Its axis-pair contrast is $0.078475$, compared with $0.083251$, or $-3.14$ theoretical binomial standard errors. Both observations are retained. A separately seeded diagnostic replication with one million trials gives $0.590442$, a 0.90-standard-error deviation. It was performed in response to the discrepancy and is reported separately; it is not substituted into the primary estimates. Including it, the archive contains 4.6 million source trials and 9.2 million earlier probe records. These simulations characterize the stipulated model, not a laboratory realization.

Periodic quadrature independently integrates the original endpoint kernels and probe responses on grids of 512, 2,048, 8,192, and 32,768 points. There are 60 equal-width and 32 unequal-width parameter cases per grid. At 512 points the largest equal-width relative normalizer error is $8.02\times10^{-4}$, demonstrating under-resolution of a narrow source. At grids of at least 2,048 points the tested second-moment errors are below $9.86\times10^{-16}$ and the largest terminal-pair error is $3.00\times10^{-15}$. The largest absolute normalizer residual in those grids is $7.28\times10^{-14}$, with relative residual $4.58\times10^{-15}$. These are observed floating-point comparisons, not rigorous discretization bounds or a common accuracy guarantee.

For the rigidity theorem we integrate $l_k(\lambda)=1/2+0.3\tanh(k\cos\lambda)$ for eight values $k\in\{1/2,1,2,4,8,16,32,64\}$ and compare with a constant response and $l_{\rm sign}=1/2+0.3\operatorname{sgn}(\cos\lambda)$. The 65-axis sweep in Fig.~\ref{fig:rigidity} uses unequal widths $g_A=0.2$, $g_B=0.7$ and $x=0.17$. Smooth responses retain a nonzero two-use variation, decreasing to $6.71\times10^{-4}$ at $k=64$. The sign exception is checked for lengths one through eight, grouped by plus-count, giving 2,860 ordered-string probability evaluations; its largest residual is $4.44\times10^{-16}$. This grouping uses exchangeability and is not enumeration of every binary string.

Finally, quantum-instrument matrix experiments use 20 seeded mixed qubit states, weak $X/Z$ measurements, and readable or unread amplitude damping. Two adaptive policies of horizon four are applied to each of the 20 states, giving 40 policy--state runs. Together with two alternative suffixes for each of 180 state--prefix cases, they produce 1,240 recorded branches. The largest forward/backward probability difference is $2.22\times10^{-16}$ and the prefix-preservation residual is $4.44\times10^{-16}$. The experiment also evaluates Eq.~\eqref{eq:calibration}, whose marginal agreement is imposed algebraically. It does not test a hidden factorization of that calibration. Appendix~\ref{app:diagnostics} presents the width-sensitivity, response-rigidity, quadrature-convergence, and sampling-residual plots. The underlying tables and code are supplied with the source package.

\section{Discussion}\label{sec:discussion}
The defect space and observation space address distinct requirements. The first certifies that later continuations preserve the weights of earlier records; the second tests whether context-dependent hidden preparations can be distinguished by the declared readouts. Their separate roles persist even when every instrument is normalized. Outcome-wise closure is essential for adaptive control and for readable environments, as the fixed-sequence counterexample demonstrates.

Quantum instruments provide a useful comparison~\cite{Chiribella2009}. Let $\omega\geq0$ with $\Tr\omega=1$, and let $\{\mathcal J_{u,o}\}_o$ be completely positive operations whose sum preserves trace. On a causal branch $h$ let $\mathcal J_h$ be their ordered composition and $E_h=\mathcal J_h^*(I)$, where $\Tr[\mathcal J(\omega)E]=\Tr[\omega\mathcal J^*(E)]$. Its probability is $n_h=\Tr(\omega E_h)$. Summing outcomes of the final operation replaces the sum of adjoint effects by $I$; repeating backwards proves normalization and preservation of any earlier prefix, including under classical feed-forward. This is the ordered-operator analogue of mass preservation and outcome-wise cone closure.

The process-tensor formalism describes multi-time experiments with environmental memory and correlated control operations~\cite{Pollock2018Framework}. Quantum Markov order can depend on the instruments used to probe a process~\cite{Taranto2019}, reinforcing the need to specify the accessible operations. Our criteria address normalization and context-independence of classical hidden-state records; they do not classify the Markov order of a quantum process.

A branchwise calibration can impose these observable probabilities on arbitrary hidden-history weights. Assign a nonnegative raw trajectory measure $R_h(dz)$ with finite mass $d_h=\int R_h(dz)$. Provided $d_h>0$ whenever $n_h>0$, define
\begin{equation}
 C_h=\frac{n_h}{d_h},\qquad P(h,dz)=C_hR_h(dz).
 \label{eq:calibration}
\end{equation}
Omit branches with $d_h=n_h=0$ and set $C_h=0$ when $n_h=0<d_h$. Integration gives $P(h)=n_h$ by construction. This is an operational calibration, not a derived detector mechanism. The complete-branch factor $C_h$ need not factor into a fixed family of hidden-state instruments. Applying the matrix test only to kernels generating $R_h$ would test the uncorrected law. A certificate for the calibrated construction requires sufficient hidden states at each cut, a joining rule, and all accessible refinements.

A backward-propagated likelihood or effect is not automatically a readable physical variable. If it is made accessible, its response, back-action, support, and influence on later control become additional components of the model. Likewise, a normalized reset can expose inadmissible subensembles to subsequent raw-weight devices. Calibration of coarse records does not establish consistency for either extension.

Adaptive conditional reconstruction is another model prescription. Wharton and Argaman's discussion~\cite{Wharton2020} includes a law of the form $\sum_A P'_a(A)P_{a,f(A)}(Q\mid A)$: a marginal law for an earlier outcome is combined with a conditional hidden law for the feedback-selected setting. This is not generally the same prescription as fixed-kernel multiplication and one global normalization in Eq.~\eqref{eq:archive}. The two constructions should be compared through their explicit source assignments and accessible records, rather than identified merely because both allow later settings to enter hidden distributions.

The passive spin extension fails the record-independence requirement under its stated unchanged-source assumptions. This finding neither excludes active or context-sensitive completions nor derives such a completion. The finite-state results apply to finite causal trees within a specified interaction family. Spatial no-signalling requires a locality specification and distant-marginal tests; causal loops and unrestricted cyclic wiring are outside this class. Joint measurements on independently sourced systems remain additional model-construction problems.

\paragraph{Data and code availability.}
The accompanying source package contains the Python simulators, exact finite-state calculations, per-trial sampled records, tabulated quadrature results, and figure-generation code, together with the LaTeX sources and vector figures required for compilation. Appendix~\ref{app:diagnostics} presents the sensitivity and convergence diagnostics.

\paragraph{Author contributions and use of AI tools.}
The author used OpenAI's ChatGPT to assist with the implementation of numerical verification code and the preparation of the manuscript. The author takes responsibility for the final arguments, computations, interpretations, and conclusions.

\bibliographystyle{quantum}
\bibliography{references}

\onecolumn
\raggedbottom
\appendix
\section{Source normalization and circular moments}\label{app:moment}
This derivation also fixes the normalization used by the numerical sampler. The wrapped-Cauchy convolution identity $p_{g_A}*p_{g_B}=p_{g_A+g_B}$ follows by multiplying Fourier coefficients $e^{-g|k|}$. Its antipodal version gives $Z_{xy}=2R_{g_A+g_B}(x-y)$ in Eq.~\eqref{eq:rho}. Hence the terminal pair probabilities are
\begin{equation}
 P(a,b\mid x,y)=\frac{p_{g_A+g_B}(\theta_{x,a}-\theta_{y,b}+\pi)}{Z_{xy}}.
 \label{eq:terminalpair}
\end{equation}
For equal widths this simplifies to
\begin{equation}
 P(a,b\mid x,y)=\tfrac14[1-ab\sech(2g)\cos(x-y)].
 \label{eq:terminalsimple}
\end{equation}
Finite positive $g$ therefore gives a visibility-reduced correlation, approaching $-\cos(x-y)$ only as $g\to0$. The numerical experiments test the finite-width model actually specified here.

For an independent moment derivation, define $\alpha=e^{-2g_A}$, $\beta=e^{-2g_B}$, $a=e^{2ix}$, $b=e^{2iy}$, and $t=a/b$. Apart from a common factor $\pi^{-2}$, the zeroth Fourier coefficient of the product in Eq.~\eqref{eq:rho} is
\begin{equation}
 D=\frac{1-(\alpha\beta)^2}{(1-\alpha\beta t)(1-\alpha\beta/t)}.
 \label{eq:Dfourier}
\end{equation}
The integral containing an additional factor $e^{2i\lambda}$ is obtained by selecting terms whose Fourier indices sum to $-1$. Splitting the resulting series at indices 0 and $-1$ gives the coefficient
\begin{equation}
 N=\frac{\beta b}{1-\alpha\beta/t}
   +\frac{\alpha a}{1-\alpha\beta t}.
\end{equation}
Therefore
\begin{align}
 \E_{xy}e^{2i\lambda}
 &=\frac{\alpha(1-\beta^2)a+\beta(1-\alpha^2)b}
 {1-\alpha^2\beta^2}\nonumber\\
 &=\frac{\sinh(2g_B)e^{2ix}+\sinh(2g_A)e^{2iy}}
 {\sinh[2(g_A+g_B)]}.
 \label{eq:unequalmoment}
\end{align}
All series converge absolutely for positive widths. Equal widths give Eq.~\eqref{eq:moment}; coincident axes require no limiting singularity.

For comparison, a contour derivation of the equal-width identity sets $q=e^{-2g}$, $z=e^{2i\lambda}$, $a=e^{2ix}$, and $b=e^{2iy}$. Up to a common factor the zeroth moment has integrand
\begin{equation}
 \frac{z\,dz}{(z-qa)(z-qb)(1-qz/a)(1-qz/b)}.
\end{equation}
For distinct $a,b$ the two interior residues are proportional to
\begin{align}
 A&=\frac{a}{(a-b)(1-q^2)(1-q^2a/b)},\nonumber\\
 B&=\frac{b}{(b-a)(1-q^2)(1-q^2b/a)}.
\end{align}
The first-$z$ moment divided by the zeroth is $(qaA+qbB)/(A+B)=q(a+b)/(1+q^2)$, again Eq.~\eqref{eq:moment}. Coincident axes follow by continuity for $g>0$.

\section{Measurable state spaces}\label{app:measurable}
For measurable hidden spaces, matrices are replaced by positive kernels, and measures are paired with measurable functions. With finite outcome sets and all required integrals finite, Eq.~\eqref{eq:effect} and its defect induction remain valid on algebraic spans of finite-policy effects. The word-effect proof likewise characterizes finite record strings. There need not be a finite basis or a uniform witness-depth bound. In particular, an environment whose relevant state space grows without bound does not inherit a fixed $n$-step certificate.

Continuous outcome spaces additionally require a specified function space, measurable kernel families, and a determining class of record events, with sums replaced by integrals. Such analytic closure cannot be inferred from a finite discretization. The spin-source rigidity theorem is proved directly by Fourier uniqueness; numerical integration only tests its stated examples.
\clearpage
\section{Numerical sensitivity and convergence diagnostics}\label{app:diagnostics}
Figures~\ref{fig:width}--\ref{fig:residuals} present the sensitivity and convergence diagnostics for the experiments in Section~\ref{sec:experiments}, using the archived observations and quadrature tables.

\noindent\appendixfigure{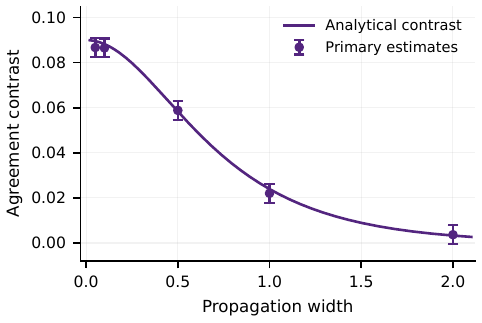}{Agreement contrast as a function of propagation width at $x=0$ and $\epsilon=0.6$. The curve is $\Delta_{\rm agree}=\epsilon^2\sech(2g)/4$. Each point combines the independent primary conditions at $y=0$ and $y=\pi/2$, with 100,000 trials per condition. Bars are pointwise 95\% difference intervals. The diagnostic replication is excluded.}{fig:width}\hfill
\appendixfigure{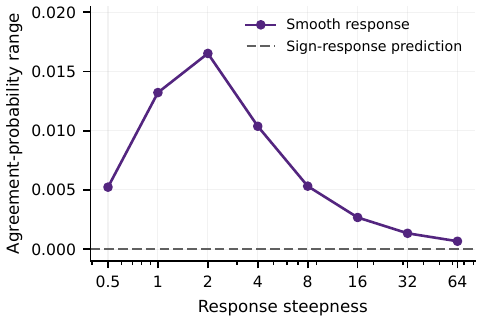}{Two-probe agreement variation for $l_k(\lambda)=1/2+0.3\tanh(k\cos\lambda)$. The ordinate is the maximum minus the minimum agreement probability over 65 later axes, with $g_A=0.2$, $g_B=0.7$, and $x=0.17$. Circles denote quadrature results for eight steepness values; connecting lines are guides on a logarithmic horizontal scale. The dashed line is the analytically zero variation of $l_{\rm sign}=1/2+0.3\operatorname{sgn}(\cos\lambda)$.}{fig:rigidity}

\vspace{1.5em}
\noindent\appendixfigure{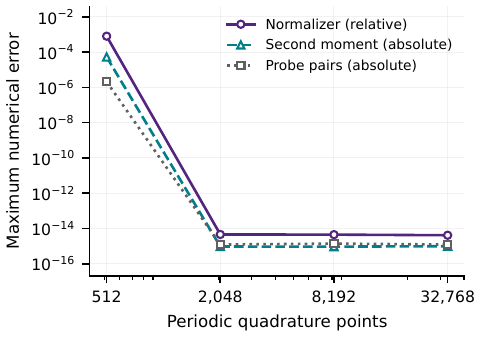}{Periodic-quadrature errors, maximized over the 60 equal-width parameter cases at each grid size. The three series show the relative normalizer error, absolute complex second-moment error, and maximum absolute error among the four probe-pair probabilities. Both axes are logarithmic. At 512 points the narrowest sources are under-resolved; at the larger grids the comparisons are dominated by floating-point precision. Lines connect tested grids and are not rigorous bounds. Unequal-width and terminal-pair comparisons are tabulated but not plotted.}{fig:convergence}\hfill
\appendixfigure{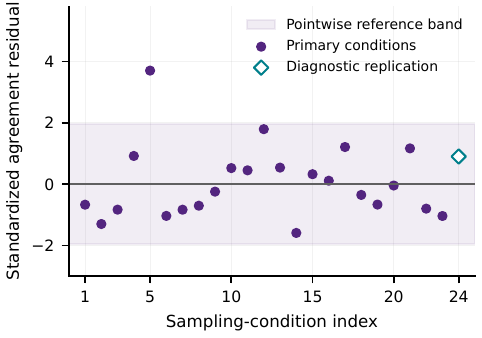}{Standardized agreement residuals $z=(\hat p-p)/\sqrt{p(1-p)/n}$, where $p$ is the analytical agreement probability. Circles show the 23 primary conditions in their archived order: five axis conditions, eight additional strength conditions, and ten width conditions. The open diamond at index 24 is the separately seeded diagnostic replication. The shaded band $|z|\leq1.96$ is a pointwise normal reference, without simultaneous coverage. The primary orthogonal-axis condition at index 5 remains visible at $z=3.70$.}{fig:residuals}

\end{document}